%% file: main.tex
\documentclass{article}
\usepackage[margin=1in]{geometry}

\usepackage[font=small,labelfont=bf]{caption}
\usepackage{booktabs} 
\usepackage{libertine}
\usepackage{tcolorbox}

\newcommand{\approach}[1]{{$\mathsf{ #1}$}}

\usepackage{amsmath,amssymb,amsfonts,amsthm,framed}
\usepackage{latexsym}
\usepackage{hyperref}

\DeclareMathOperator{\polylog}{polylog}

\newcommand{\tO}{\tilde{O}}

\renewcommand{\implies}{\mbox{ implies }}

\newcommand{\prob}[1]{ \Pr \left [ #1 \right ]}
\newcommand{\size}[1]{  \left | #1 \right |}

\newcommand{\COMMENTED}[1]{{}}

\newcommand{\expec}[1]{ \textup{E}\left [ #1 \right ]}

\newcommand\given[1][]{\:#1\vert\:}

\newcommand{\paren}[1]{\left ( #1 \right )}

\newcommand{\junk}[1]{}

\newcommand{\fhat}{\hat{f}}

\newcommand{\aand}{\mbox{ and }}

\newcommand{\etal}{{et al.~}}

\newcommand{\tforall}{\mbox{ for all }}

\theoremstyle{plain}

\newtheorem{theorem}{Theorem}[section]
\newtheorem{lemma}[theorem]{Lemma}
\newtheorem{corollary}[theorem]{Corollary}

\usepackage{enumitem}
\setlist{leftmargin=5.5mm}

\usepackage{framed,enumitem}

\title{Finding Subcube Heavy Hitters in Analytics Data Streams}
\date{}

\author{Branislav Kveton\thanks{The authors are listed alphabetically.} \\ Adobe Research \\ kveton@adobe.com  \and S. Muthukrishnan \\ Rutgers University \\ muthu@cs.rutgers.edu \and Hoa T. Vu\thanks{Part of this work was done at Adobe Research.} \\ University of Massachusetts \\ hvu@cs.umass.edu \and Yikun Xian \\ Rutgers University \\ siriusxyk@gmail.com }

\begin{document}

\begin{titlepage}

\maketitle
\sloppy
\input{abstract}
\end{titlepage}

%
%
%
%

\input{introduction}
\input{sampling}
\input{independence}
\input{naivebayes}
\input{experiments}
\input{concl}

\bibliographystyle{plain}
\bibliography{mybib} 

\end{document}

%% file: abstract.tex

\begin{abstract}

Modern data streams typically have high dimensionality. For example, digital analytics streams consist of  user online activities (e.g., web browsing activity, commercial site activity, apps and social behavior, and response to ads). An important problem is to find frequent joint values (heavy hitters) of subsets of dimensions. 

Formally, the data stream consists of $d$-dimensional items and a {\em $k$-dimensional subcube} $T$ is a subset of $k$ distinct coordinates. Given a theshold $\gamma$, a {\em subcube heavy hitter query} ${\rm Query}(T,v)$ outputs YES if  $f_T(v) \geq \gamma$ and NO if $f_T(v) < \gamma/4$ where $f_T$ is the ratio of the number of stream items whose coordinates $T$ have joint values $v$. 
The {\em all subcube heavy hitters query} ${\rm AllQuery}(T)$ outputs all joint values $v$ that return YES to ${\rm Query}(T,v)$.  The problem is to answer these queries correctly for all $T$ and $v$. 

We present a simple one-pass sampling algorithm to solve the subcube heavy hitters problem in $\tilde{O}(kd/\gamma)$ space. $\tilde{O}(\cdot)$ suppresses polylogarithmic factors. This is optimal up to polylogarithmic factors based on the lower bound of Liberty \etal \cite{LibertyMTU16} In the worst case, this bound becomes $\Theta(d^2/\gamma)$ which is prohibitive for large $d$. 

Our main contribution is to circumvent this quadratic bottleneck via a model-based approach. In particular, we assume that the dimensions are related to each other via the Naive Bayes model. 
We present a new two-pass, $\tilde{O}(d/\gamma)$-space algorithm for our problem, and a fast algorithm for answering ${\rm AllQuery}(T)$ in $\tO((k/\gamma)^2)$ time. 

We demonstrate the effectiveness of our approach on a synthetic dataset as well as real datasets from Adobe and Yandex.  Our work shows the potential of model-based approach to data streams.
\end{abstract}

%% file: introduction.tex
\section{Introduction}

We study the problem of finding heavy hitters in high dimensional data streams. Most companies see transactions with items sold, time, store location, price, etc. that arrive over time. Modern online companies see streams of user web activities that typically have components of user information including ID (e.g. cookies), hardware (e.g., device), software (e.g., browser, OS), and  contents such as web properties, apps. Activity streams also include events (e.g., impressions, views, clicks, purchases) and  event attributes (e.g., product id, price, geolocation, time). 
Even classical IP traffic streams have many dimensions including source and destination IP addresses, port numbers and other features of an IP connection such as application type. Furthermore, in applications such as Natural Language Processing, streams of documents can be thought of as streams of a large number of bigrams or multi-grams over word combinations~\cite{GoyalDV09}. As these examples show, analytics data streams with 100's and 1000's of dimensions arise in many applications. Motivated by this, we study the problem of finding {\em heavy hitters} on data streams focusing on $d$, the number of dimensions, as a parameter. Given $d$ one sees in practice, $d^2$ in space usage is prohibitive, for solving the heavy hitters problem on such streams.

Formally, let us start with a one-dimensional stream of items $x_1, \ldots x_m$  where each $x_i \in [n] := \{ 1,2,\ldots,n\}$. We can look at the {\em count} $c(v)=  \size{ \{ i: x_i = v \} }$ or the {\em frequency ratio} $f(v) = {c( v)}/{m}$. A {\em heavy hitter} value $v$ is one with 
$c(v) \geq \gamma m$ or equivalently $f(v) \geq \gamma$, for some constant $\gamma$. The standard {\em data stream model} is that we maintain data structures of size $\polylog(m,n)$ and determine if $v$ is a heavy hitter with probability of success at least $3/4$, that is, if $f(v) \geq \gamma$ output YES and output NO if $f(v) < \gamma/4$ for all $v$.\footnote{The gap constant 4 can be narrowed arbitrarily and the success probability can be amplified to $1-\delta$ as needed, and we omit these factors in the discussions.} We note that if $\gamma/4 \leq f(v) < \gamma$, then either answer is acceptable.
 
Detecting heavy hitters on data streams is a fundamental problem that arises in guises such as finding elephant flows and network attacks in networking, finding hot  trends in databases, finding frequent patterns in data mining,  finding largest coefficients in signal analysis, and so on.  Therefore, the heavy hitters problem has been studied extensively in theory, databases, networking and signal processing literature.  See~\cite{graham} for an early survey and ~\cite{Woodruff16} for a recent survey.
 
 \noindent
\paragraph{Subcube heavy hitter problems} 
 Our focus is on modern data streams such as in analytics cases, with $d$ dimensions, for large $d$. The data stream consists of $d$-dimensional 
 items $x_1,\cdots, x_m$. In particular,  
\[x_i  = (x_{i,1},\ldots,x_{i,d}) \mbox{ and each } x_{i,j} \in [n] ~.\]
A {\em $k$-dimensional subcube} $T$ is a subset of $k$ distinct coordinates $\{ T_1,\cdots,T_k \} \subseteq [d]$. We refer to the joint values of the coordinates $T$ of $x_i$ as $x_{i,T}$. 

The number of items whose coordinates $T$ have joint values $v$ is denoted by $c_T(v)$, i.e., $c_T(v) = \size{ \{ i: x_{i,T} = v \} }.$ Finally, we use $X_T$ to denote the random variable of the joint values of the coordinates $T$ of a random item. We  have the following relationship
\[
f_T(v) := \prob{X_T = v} = \frac{c_T(v)}{m}  ~.
\]
For a single coordinate $i$, we slightly abuse the notation by using $f_i$ and $f_{\{i\}}$ interchangeably. For example, $f_{T_i}(v)$ is the same as $f_{ \{ T_i \} }(v)$. Similarly, $X_i$ is the same as $X_{\{ i \}}$.

We are now ready to define our problems. They take $k, \gamma$ as parameters and the stream as the input and build data structures to answer:
\begin{itemize}
\item
{\em Subcube Heavy Hitter}:  ${\rm Query}(T,v)$,  where $\size{T} = k$, and $v \in [n]^k$, returns an estimate if $f_T(v) \geq \gamma$. Specifically, output YES if  $f_T(v) \geq \gamma$ and NO if $f_T(v) < \gamma/4$. If $\gamma/4 \leq f_T(v) < \gamma$, then either output is acceptable. 
The required success probability {\em for all} $k$-dimensional subcubes $T$ and $v \in [n]^k$ is at least 3/4.

\item
{\em All Subcube Heavy Hitters}:  ${\rm AllQuery}(T)$ outputs all joint values $v$ that return YES to ${\rm Query}(T,v)$. This is conditioned on the algorithm used for ${\rm Query}(T,v)$. 
\end{itemize}
It is important to emphasize that the stream is presented (in a single pass or constant passes) to the algorithm before the algorithm receives any query.

Subcube heavy hitters are relevant wherever one dimensional heavy hitters have 
found applications: combination of source and destination IP addresses forms the subcube heavy hitters that detect network attacks; combination of stores, sales quarters and nature of products forms the subcube heavy hitters that might be the pattern of interest in the data, etc. Given the omnipresence of multiple dimensions in digital analytics, arguably, subcube heavy hitters limn the significant data properties far more than the single dimensional view. 

\paragraph{Related works} 
The problem we address is directly related to frequent itemset mining studied in the data mining community. 
In frequent itemset mining, each dimension is binary ($n=2$), and we consider ${\rm Query}(T,v)$ where $v= (1,\ldots,1) := {\bf U}_k$. 
It is known that counting all maximal subcubes $T$ that have a frequent itemset, i.e., $f_T({\bf U}_k) \geq \gamma$, is $\#P$-complete \cite{Yang04}. Furthermore,  finding even a single $T$ of maximal size such that $f_T({\bf U}_k) \geq \gamma$ is NP-hard \cite{HamiltonCW06,LibertyMTU16}. 
Recently,  Liberty \etal showed that any constant-pass streaming algorithm answering ${\rm Query}(T,{\bf U}_k)$ requires 
$\Omega(kd/\gamma \cdot \log(d/k))$  bits of memory \cite{LibertyMTU16}. In the worst case, this is $\Omega(d^2/\gamma)$ for large $k$, ignoring the polylogarithmic factors. 
For this specific problem, sampling algorithms will nearly meet their lower bound for space. Our problem is more general, with arbitrary $n$ and $v$.

\paragraph{Our contributions}
Clearly, the case $k=1$ can be solved by building one of the many known single dimensional data structures for the heavy hitters problem on each of the $d$ dimension; the $k=d$ case can be thought of as a giant single dimensional problem by linearizing the space of all values in $[n]^k$; for any other $k$, there are ${d \choose k}$ distinct choices for subcube $T$, and these could be treated as separate one-dimensional problems by linearizing each of the subcubes.  In general, this entails ${d \choose k}$ and $\log (n^d)$ cost in space or time bounds over the one-dimensional case, which we seek to avoid. Also, our problem can be reduced to the binary case by unary encoding each dimension by $n$ bits, and solving frequent itemset mining:  the query then has $kn$ dimensions. The resulting bound will have an additional $n$ factor which is large. 

First, we observe that the reservoir sampling approach ~\cite{Vitter85}
solves subcube heavy hitters problems more efficiently compared to the approaches mentioned above. Our analysis shows that the space we use is within polylogarithmic factors of the lower bound shown in~\cite{LibertyMTU16} for binary dimensions and query vector ${\bf U}_k$, which is a special case of our problem. Therefore, the 
subcube heavy hitters problem can be solved using  $\tO(kd/\gamma)$ space. 
However, this is $\Omega(d^2)$ in worst case. 

Our main contribution is to avoid this quadratic bottleneck for finding subcube heavy hitters. 
We adopt the notion that there is an underlying probabilistic model behind the data, and
in the spirit of the Naive Bayes model, we assume that the dimensions are nearly (not exactly) mutually independent given an observable latent dimension. This could be considered as a low rank factorization of the dimensions.
In particular, one could formalize this assumption by bounding the total variational distance between the data's joint distribution and that derived from the
Naive Bayes formula. 
  This assumption is common  in statistical data analysis and highly prevalent in machine learning. 
Following this modeling, we make two main contributions:
\begin{itemize}
\item
We present a two-pass,  $\tO(d/\gamma)$-space streaming algorithm for answering ${\rm Query}(T,v)$.  This improves upon the $kd$ factor 
in the space complexity from sampling, without assumptions, to just $d$ with the Naive Bayes assumption, which would make this algorithm
 practical for large $k$. 
Our algorithm uses sketching in each dimension in one pass to detect heavy hitters, and then needs a second pass to precisely estimate their frequencies. 

\item
We present a fast algorithm for answering ${\rm AllQuery}(T)$ in $\tO((k/\gamma)^2)$ time.
The naive procedure would take exponential time $\Omega((1/\gamma)^k)$ by considering the Cartesian product of the heavy hitters in each dimension. Our approach, on the other hand, uses the structure of the Naive Bayes assumption to iteratively construct the subcube heavy hitters one dimension at a time.
\end{itemize}

Our work develops the direction of model-based data stream analysis. Model-based data analysis has been effective in other areas. For example, in compressed sensing,  realistic signal models that include dependencies between values and locations of the signal coefficients improve upon unconstrained cases \cite{duarte2009}. In statistics, using tree constrained models of multidimensional data sometimes improves point and density estimation.  In high dimensional distribution testing, model based approach has also been studied to overcome the curse of dimensionality \cite{DaskalakisDK16}. 

In the data stream model, \cite{IndykM08,BravermanCLMO10,BravermanO10} studied the problem of testing independence. McGregor and Vu \cite{McGregorV15} studied the problem of evaluating Bayesian Networks. In another work, Kveton \etal \cite{KvetonBGTMS16}
 assumed a tree graphical model and designed a one-pass algorithm that estimates the joint frequency; their work however only solved the $k=d$ case for the  joint frequency estimation problem.
Our model is a bit different and more importantly, we solve the subcube heavy hitters problem
(addressing all the ${d \choose k}$ subcubes) which prior work does not solve.  
In following such a direction, we have extended the fundamental heavy hitters problem to
higher dimensional data. Given that many implementations already exist for the sketches we use for one-dimensional heavy hitters as a blackbox,  our algorithms are therefore easily implementable. 

\medskip
\paragraph{Background on the Naive Bayes model and its use in our context.} The Naive Bayes Model \cite{russell10artificial} is a Bayesian network over $d$ features $X_1, \dots, X_d$ and a class variable $Y$. This model represents a joint probability distribution of the form
\begin{align*}
&  \prob{X_1 = x_1, \dots, X_d = x_d, Y = y} \\
  = & \prob{Y = y} \prod_{j = 1}^d \prob{X_j = x_j \mid Y = y}\,,
\end{align*}
which means that the values of the features are conditionally independent given the value of the class variable. 
The simplicity of the Naive Bayes model makes it a popular choice in text processing and information retrieval \cite{LewisYRL04,manning08information}, with state-of-the-art performance in spam filtering \cite{Androutsopoulos00}, text classification \cite{LewisYRL04}, and others. 

\junk{
For the discussion purpose in this paper, we focus on text classification where the {\em features} are individual words in text, which is also known as the bag-of-words model; and the {\em class label} summarizes one aspect of the text, such as its topic of the text or the content type \cite{LewisYRL04}. Therefore, the features are discrete variables with high cardinality and the class is the topic of the text. The number of topics is typically on the order of tens \cite{LewisYRL04}. Many probabilistic queries of interest are in the form of the Naive Bayes model. Given such streams of text, an example query of interest is
\begin{align*}
  \{(x_{i_1}, x_{i_2}): \prob{X_{i_1} = x_{i_1}, X_{i_2} = x_{i_2}} \geq \gamma\}
\end{align*}
all bi-grams over positions $i_1$ and $i_2$ that appear frequently in the text. Such queries can, for instance, help in understanding the generative aspect of the model by retrieving the sequences of words that are likely to be generated. In this work, we show how to answer these queries in a small space for $k$-grams of any subset of up to $d$ features.
}

\paragraph{Empirical study.}
We perform detailed experimental study of subcube heavy hitters. We use a synthetic dataset where we generate data that confirms to the Naive Bayes model. We then experiment with real data sets from Yandex (Search) and Adobe (Marketing Cloud) which give multidimensional analytics streams. We experiment with the reservoir sampling based algorithm as a benchmark that works without modeling assumptions, and our two-pass subcube heavy hitters algorithm that improves upon it for data that satisfies the model. We also adopt our approach to give a simpler one-pass algorithm for which theoretical guarantees is weaker. Our experiments show substantial improvement of the model-based algorithms over the benchmark for synthetic as well as real data sets, and further show the benefits of the second pass. 

%% file: sampling.tex
\section{The Sampling Algorithm} \label{sec:sampling}

In this section, we show that sampling solves the problem efficiently compared to  running one-dimensional heavy hitters algorithms for each of ${d \choose k}$ $k$-dimensional subcubes independently. It also matches the lower bound in \cite{LibertyMTU16} up to polylogarithmic factors.

\paragraph{Algorithm details.} The algorithm samples $m' = \tO( \gamma^{-1} k d )$ random items $z_1,\ldots,z_{m'}$ from the stream using Reservoir sampling~\cite{Vitter85}. Let $S = \{ z_1,\ldots,z_{m'}\}$ be the sample set.
 Given ${\rm Query}(T,v)$, we output YES if and only if the sample frequency of $v$, denoted by $\fhat_T(v)$, is at least $\gamma/2$. Specifically, 
 \[\fhat_T(v) := \frac{| \{ x_i : x_i \in S \aand  x_{i,T} = v \}  |}{m'} ~.\] 
 
For all subcubes $T$ and joint values $v$ of $T$, the expected sample frequency $\fhat_T(v)$ is $ f_T(v)$. Intuitively, if $v$ is a frequent joint values, then its sample frequency $\fhat_T(v)  \approx f_T(v)$; otherwise, $\fhat_T(v)$ stays small. 

Let us fix a $k$-dimensional subcube $T$ and suppose that for all $v \in [n]^{k}$, we have   
\begin{align} \label{eq:approximation}
\fhat_T (v) = f_T(v) \pm  \frac{\max\{\gamma,f_T(v)\}}{4} ~.
\end{align}

It is then straightforward to see that if $f_T(v) < \gamma/4$, then  $\fhat_T (v) < \gamma/4+ \gamma/4 = \gamma/2$. Otherwise, if $f_T(v) \geq \gamma$, then  $\fhat_T (v) \geq 3f_T(v)/4 \geq 3\gamma/4 > \gamma/2$. Hence, we output YES for all $v$ where $\fhat_T (v) \geq \gamma/2$, and output NO otherwise. 

\begin{lemma} (Chernoff bound) \label{thm:chernoff}
Let $X_1,\cdots, X_n$ be independent or negatively correlated binary random variables.
Let $X = \sum_{i=1}^n X_i$ and $\mu = \expec{X}$. Then,
 \[
 \prob{ | X - \mu | \geq \epsilon \mu } \leq \exp(-\min\{ \epsilon^2 ,\epsilon \} \mu /3) ~.
 \]
\end{lemma}

Recall that $S = \{ z_1,z_2,\ldots,z_{m'} \}$ is the sample set returned by the algorithm. For a fixed $v \in [n]^k$, we use  $Z_i$ as the indicator variable for the event $z_{i,T}=v$. Since we sample without replacement, the random variables $Z_i$ are negatively correlated. The following lemma shows that Eq. \ref{eq:approximation} holds for all $v$ and $k$-dimensional subcubes $T$ via Chernoff bound.



\begin{lemma}
For all $k$-dimensional subcubes $T$ and joint values $v \in [n]^k$, with  probability at least 0.9, 
\[
\fhat_T (v) = f_T(v) \pm  \frac{ \max\{\gamma,f_T(v)\}}{4} ~. 
\]
\end{lemma}

\begin{proof}
Let $m' = c \gamma^{-1} \log (d^k \cdot n^k)$ for some sufficiently large constant $c$. We first consider a fixed $v \in [n]^k$ and define the random variables $Z_i$ as above, i.e.,  $Z_i = 1$ if $z_{i,T} = v$. Suppose $f_T(v) \geq \gamma$. Appealing to Lemma \ref{thm:chernoff}, we have
\begin{align*}
& \prob{\left| \paren{\sum_{i=1}^{m'} \frac{Z_i}{m'}} - f_T(v) \right| \geq \frac{f_T(v)}{4} }  \\
 = & \prob{\left| \fhat_T(v) - f_T(v) \right| \geq \frac{f_T(v)}{4} } \\ 
 \leq &  \exp \left( -\frac{f_T(v) m'}{3 \times 16}  \right) \leq \frac{1}{10 d^k n^k } ~.
\end{align*}
On the other hand, if  $f_T(v) < \gamma/4$, then
\begin{align*}
\prob{\left| \fhat_T(v) - f_T(v) \right| \geq  \frac{\gamma}{4}} & \leq  \exp \left( -\left( \frac{\gamma}{4 f_T(v)} \right) f_T(v) \frac{m'}{3} \right) \\ &   \leq  \frac{1}{10 d^k n^k } ~.
\end{align*}
Therefore, by taking the union bound over all ${d \choose k} \cdot n^k \leq d^k \cdot n^k$ possible combinations of $k$-dimensional subcubes and the corresponding joint values $v \in [n]^k$, we deduce the claim.
\end{proof}

We therefore could answer all ${\rm Query}(T,v)$ correctly with probability at least 0.9 for all joint values $v \in [n]^k$ and $k$-dimensional subcubes $T$. Because storing each sample $z_i$ requires $\tO(d)$ bits of space, the algorithm uses $\tO(dk \gamma^{-1})$ space. We note that answering ${\rm Query}(T,v)$ requires computing $\fhat_T(v)$ which takes $O(\size{S})$ time. We can answer ${\rm AllQuery}(T)$ by computing  $\fhat_T(v)$ for all joint values $v$ of coordinates $T$ that appear in the sample set which will take $O(\size{S}^2)$ time. We summarize the result as follows.
\begin{theorem}
There exists a 1-pass algorithm that uses $\tO(d k \gamma^{-1} )$ space and solves $k$-dimensional  subcube heavy hitters. Furthermore, ${\rm Query}(T,v)$ and  ${\rm AllQuery}(T)$ take $\tO(d k \gamma^{-1})$ and $\tO \paren{\paren{d k \gamma^{-1}}^2 }$ time respectively.
\end{theorem}

%% file: independence.tex

\section{The  Near-Independence Assumption} \label{sec:near-independence}

\paragraph{The near-independence assumption. }Suppose the random variables representing the dimensions $X_1,X_2,\ldots,X_d$ are {\em near} independent. We show that there is a 2-pass algorithm that uses less space and has faster query time. At a high level, we make the assumption that the joint probability is approximately factorized
\[
f_{\{1,\ldots,d\}}(v) \approx f_1 (v_1) f_2(v_2) \cdots f_d (v_d) ~.
\]

More formally, we assume that the total variation distance is bounded by a small quantity $\alpha$. Furthermore, we assume that $\alpha$ is reasonable with respect to $\gamma$ that controls the heavy hitters.  For example, $\alpha \leq \gamma/10$ will suffice. 

\begin{tcolorbox}
The formal {\em near-independence} assumption is as follows: There exists $\alpha \leq  \gamma/10$ such that for all subcubes $T$,
\[
\max_{v \in [n]^{\size{T}}} \left| f_{T}(v) - \prod_{i=1}^{\size{T}} f_{T_i} (v_i)  \right| < \alpha ~.
\]
\end{tcolorbox}
We observe that: 
\begin{itemize}
\item If $f_T(v) \geq \gamma$, then 
\[ \prod_{i=1}^{\size{T}} f_{T_i} (v_i)  \geq f_T(v)-\gamma/10 > \gamma/2~.\] 
\item If $f_T(v) < \gamma/4$, then  
\[\prod_{i=1}^{\size{T}} f_{T_i} (v_i)  \leq  f_T(v) + \gamma/10 < \gamma/2~.\] 
\end{itemize}
Thus, it suffices to output YES to ${\rm Query}(T,v)$ if and only if the marginals product $\prod_{i=1}^{\size{T}} f_{T_i} (v_i)  \geq \gamma/2$. For convenience, let \[\lambda := \gamma/2~.\] 

\paragraph{Algorithm details.} We note that simply computing $f_i(x)$ for all coordinates $i \in [d]$ and $x \in [n]$ will need $\Omega(dn)$ space. To over come this, we make following simple but useful observation. We observe that if $v$ is a heavy hitter in the subcube $T$ and if $T'$ is a subcube of $T$, then $v_{T'}$ is a heavy hitter in the subcube $T'$. 

\begin{lemma} \label{clm:subspace}
For all subcubes $T$,  
\[ 
\prod_{i=1}^{\size{T}} f_{T_i} (v_i)  \geq  \lambda \implies \prod_{i \in \mathcal{V}} f_{T_i} (v_i)  \geq \lambda
\]
 for all $\mathcal{V} \subseteq [\size{T}]$  (i.e., $\{T_i: i \in \mathcal{V} \}$ is a subcube of $T$).
 \end{lemma} 

The proof is trivial since all $f_{T_i} (v_i) \leq 1$. Therefore, we have the following corollary.

\begin{corollary} \label{cor:correctness}
For all subcubes $T$, 
\[ 
\prod_{i=1}^{\size{T}} f_{T_i} (v_i)  \geq  \lambda  \implies f_{T_i} (v_i)  \geq \lambda \tforall i \in [\size{T}] ~.
\]
\end{corollary} 

We therefore only need to compute $f_i(x)$ if $x$ is a heavy hitter in coordinate $i$. To this end, for each coordinate $i \in [d]$, by using (for example) Misra-Gries algorithm \cite{MisraG82} or Count-Min sketch \cite{CormodeM04},  we can find a set $H_i$ such that if $f_i(x) \geq \lambda/2$, then $x \in H_i$ and if $f_i(x) < \lambda/4$, then $x \notin H_i$. In the second pass, for each $x \in H_i$, we compute $f_{i}(x)$ exactly to obtain 
\[S_i := \{ x \in [n]:f_{i}(x) \geq \lambda \}  ~. \]

We output YES to ${\rm Query}(T,v)$ if and only if all $v_i \in S_{T_i}$ and $ \prod_{i=1}^{\size{T}} f_{T_i} (v_i)  \geq \lambda$. Note that if $v \in S_i$, then $f_i(v)$ is available to the algorithm since it is computed exactly in the second pass. The detailed algorithm is as follows.
\begin{tcolorbox}
\begin{enumerate}
\item First pass: For each coordinate $i \in [d]$, use Misra-Gries algorithm to find $H_i$.
\item Second pass: For each coordinate $i \in [d]$, compute $f_i (x)$ exactly for each $x \in H_i$ to obtain $S_i$. 
\item Output YES to ${\rm Query}(T,v)$ if and only if $v_i \in S_{T_i}$ for all $i \in [\size{T}]$ and 
\[ \prod_{i=1}^{{\size{T}}} f_{T_i} (v_i) \geq \lambda~.\]
\end{enumerate}
\end{tcolorbox}

\begin{theorem}
There exists a 2-pass algorithm that uses $\tO(d \gamma^{-1})$ space and solves subcube heavy hitters under the near-independence assumption. The time to answer ${\rm Query}(T,v)$ and  ${\rm AllQuery}(T)$  are  $\tO(k)$ and $\tO(  k \gamma^{-1})$  respectively where $k$ is the dimensionality of $T$.
\end{theorem}

\begin{proof}
The first pass uses $\tO(d \lambda^{-1})$ space since Misra-Gries algorithm uses $\tO( \lambda^{-1})$ space for each coordinate $i \in [d]$. Since the size of each $H_i$ is $O(\lambda^{-1})$, the second pass also uses $\tO(d \lambda^{-1})$ space. Recall that $\lambda = \gamma/2$. We then conclude that the algorithm uses $\tO(d \gamma^{-1})$ space.

For an arbitrary ${\rm Query}(T,v)$, the algorithm's correctness follows immediately from Corollary \ref{cor:correctness} and the observation that if $v_i \in S_{T_i}$, then $f_{T_i} (v_i)$ is available since it was computed exactly in the second pass. Specifically, if 
\begin{align}\label{eq:3} 
\prod_{i=1}^{{\size{T}}} f_{T_i} (v_i) \geq \lambda~,
\end{align}
then $v_i \in S_{T_i}$ for all $i \in [\size{T}]$ and we could verify the inequality and output YES. On the other hand, suppose Eq. \ref{eq:3} does not hold. Then, if $v_i \notin S_{T_i}$ for some $i$, we correctly output NO. But if all $v_i \in S_{T_i}$, then we are able to verify that the inequality does not hold (and correctly output NO).

The parameter $k$ only affects the query time. We now analyze the time to answer  ${\rm Query}(T,v)$ and  ${\rm AllQuery}(T)$ for a $k$-dimensional subcube $T$.

We can easily see that ${\rm Query}(T,v)$ takes $\tO(k )$ time as we need to check if all $v_i \in S_{T_i}$ (e.g., using binary searches) and compute $\prod_{i=1}^k f_{T_i} (v_i)$.

Next, we exhibit a fast algorithm to answer ${\rm AllQuery}(T)$. We note that naively checking all combinations $(v_1,\cdots,v_k)$ in $S_{T_1} \times S_{T_2} \times \cdots \times S_{T_k}$ takes exponential $\Omega(\gamma^{-k})$ time in the worst case. 

Our approach figures out the heavy hitters gradually and takes advantage of the near-independence assumption. In particular, define 
\[ W_j := \{ v \in [n]^j : f_{T_1}(v_1) \cdots f_{T_j}(v_j) \geq \lambda \} ~. \]

Recall that the goal is to find $W_k$. Note that $W_1 = S_1$ is obtained directly by the algorithm. We now show that it is possible to construct $W_{j+1}$ from $W_j$ in $\tO(\lambda^{-1})$ time which in turn means that we can find $W_k$ in $\tO(k \lambda^{-1})$ time. We use the notation $T_{[j]} := \{T_1,\ldots,T_j\}$ and $v_{[j]} := (v_1,v_2,\ldots,v_j)$.

We note that $\size{W_j } \leq 5/(4\lambda)$. This holds since if $y \in W_j$, then $\prod_{i=1}^j f_{T_i} (y_i) \geq \lambda $.  Appealing to the near-independence assumption, we have
\[
 f_{{T_{[j]}}} (y) \geq \prod_{i=1}^j f_{T_i} (y_i)  -\alpha \geq \lambda-\alpha \geq 4/5 \cdot \lambda ~.
\] 
For each $y \in W_j$, we collect all  $x \in S_{j+1}$ such that 
\begin{align*}
\paren{{ \prod_{i=1}^{j} f_{T_i} (y_i)}} f_{T_{j+1}}(x) \geq {\lambda}   
\end{align*}

and  put $ (y_1,\cdots,y_j,x)$ into $W_{j+1}$. Since $|W_j | \leq 5/4 \cdot \lambda^{-1}$ and $|S_{j+1}| \leq \lambda^{-1}$, this step obviously takes $O(\lambda^{-2})$ time. However, by observing that there could be at most $\lambda^{-1} \prod_{i=1}^{j} f_{T_i} (y_i)$ such $x$ for each $ y \in W_j$, the upper bound for the number of combinations of $x$ and $y$ is
\begin{align*}
  \sum_{y \in W_j}   \frac{1}{\lambda} \prod_{i=1}^{j} f_{T_i} (y_i)   & \leq  \sum_{y \in W_j}    \frac{1}{\lambda} (f_{T_{[j]}} (y)+\alpha)  \\
& =  \sum_{y \in W_j}  \frac{ \alpha}{\lambda} +  
\sum_{y \in W_j}  \frac{f_{T_{[j]}} (y)}{\lambda}  \\
& \leq  \size{ W_j}   + \frac{1}{\lambda}    \leq \frac{3}{ \lambda} ~.
\end{align*}
The last inequality follows from the assumption that $\alpha \leq \lambda/5$ and $\sum_{y \in W_j} f_{T_{[j]}}(y) \leq 1$. Thus, the algorithm can find $W_{j+1}$ given $W_j$ in $\tO(\lambda^{-1})$ time. Hence, we obtain $W_{k}$ in $\tO(k \lambda^{-1}) = \tO(k \gamma^{-1})$ time. The correctness of this procedure follows directly from Lemma \ref{clm:subspace} and induction since $v = (v_1,\ldots,v_{j+1}) \in W_{j+1}$ implies that $v_{[j]} \in W_{j}$ and $v_{j+1} \in S_{j+1}$. Thus, by checking all combinations of $y \in W_{j}$ and $x \in S_{j+1}$, we can construct $W_{j+1}$ correctly.
\end{proof}

%% file: naivebayes.tex
\section{The Naive Bayes Assumption} \label{sec:naive-bayes}
\paragraph{The Naive Bayes assumption.} In this section, we focus on the data streams inspired by the Naive Bayes model which is strictly more general than the near-independence assumption. In particular, we assume that the coordinates are near-independent given an extra $(d+1)$th observable {\em class coordinate} that has a value in $\{1,\ldots,\ell\}$. The $(d+1)$th coordinate is also often referred to as  the {\em latent coordinate}.

As in typical in Naive Bayes analysis, we assume $\ell$ is a constant but perform the calculations in terms of $\ell$ so its role in the complexity of the problem is apparent. 

Informally, this model asserts that the random variables representing coordinates $X_1,\ldots,X_d$ are near independent conditioning on a  the random variable $X_{d+1}$ that represents the class coordinate.

We introduce the following notation
\begin{align*} 
f_{T \given d+1} (v \given z) & := \frac{\size{ \{x_i: x_{i,T} = v \land x_{i,\{d+1\}} = z \}} }{ \size{ \{ x_i: x_{i, \{d+1\}} = z  \} }} \\
& = \prob{X_T = v \given X_{d+1}=z}~.
\end{align*}

In other words, $f_{T \given d+1} (v \given z)$ is the frequency of the joint values $v$ in the $T$ coordinates among the stream items where the class coordinate $d+1$ has value $z$.

\begin{tcolorbox}
The formal {\em Naive Bayes} assumption is as follows: There exists $\alpha \leq  \gamma/10$ such that for all subcubes $T$,
\[
\max_{\substack{v \in [n]^{\size{T}}}} \left| f_{T}(v) - \sum_{z\in [\ell]}f_{d+1}(z) \prod_{i=1}^{\size{T}} f_{T_i \given d+1} (v_i \given z)    \right| < {\alpha} ~.
\]
\end{tcolorbox}

\paragraph{Algorithm details.} As argued in the previous section, it suffices to output YES to ${\rm Query}(T,v)$ if and only if 
\[\sum_{z \in [\ell]} f_{d+1}(z) \prod_{i=1}^{\size{T}} f_{T_{i} \given d+1} (v_i \given z)  \geq \gamma/2 =  \lambda~.\]
However, naively computing all $f_{i \given d+1} (v \given z)$ uses $\Omega(\ell d n)$ space. We circumvent this problem by generalizing Lemma \ref{clm:subspace} as follows. If a joint values $v$ is a heavy hitter in a subcube $T$ in the Naive Bayes formula and $T'$ is a subcube of $T$, then $v_{T'}$ is a heavy hitter in the subcube $T'$. 
\begin{lemma} \label{clm:subspace2} For all subcubes $T$, 
\begin{align*}
& q(v) := \sum_{z \in [\ell]} f_{d+1}(z) \prod_{i=1}^{\size{T}} f_{T_i \given d+1} (v_i \given z ) \geq \lambda \\
& \implies
\sum_{z \in [\ell]} f_{d+1}(z) \prod_{i \in  \mathcal{V}} f_{T_i \given d+1} (v_i \given z ) \geq \lambda \end{align*}
  for all $\mathcal{V} \subseteq [{\size{T}}]$  (i.e., $\{T_i: i \in \mathcal{V} \}$ is a subcube of $T$).
\end{lemma} 
\begin{proof}
For a fixed $z$, observe that 
\begin{align*}
\sum_{y_j \in [n]} f_{T_j \given d+1}(y_j \given z) = 1 ~. 
\end{align*}
Suppose  $q(v) \geq \lambda$ and consider an arbitrary $\mathcal{V} \subseteq [\size{T}]$. We have 
\begin{align*}
 & \sum_{z \in [\ell]} f_{d+1}(z) \prod_{i \in \mathcal{V}} f_{T_i \given d+1} (v_i \given z ) \\
   = & \sum_{z \in [\ell]} f_{d+1}(z) \prod_{i \in \mathcal{V}} f_{T_i \given d+1} (v_i \given z )  \prod_{j \notin \mathcal{V}} \paren{ \sum_{y_j \in [n]} f_{T_j \given d+1}(y_j \given z)} \\
  \geq & \sum_{z \in [\ell]} f_{d+1}(z) \prod_{i \in \mathcal{V}} f_{T_i \given d+1} (v_i \given z )  \prod_{j \notin \mathcal{V}} f_{T_j \given d+1}(v_j \given z) \\
 = & \sum_{z \in [\ell]} ~ f_{d+1}(z) \prod_{i=1}^{\size{T}} f_{T_i \given d+1} (v_i \given z )  = q(v) \geq \lambda ~.
\end{align*}
An alternative proof is by noticing that $q(v)$ is a valid probability density function of $\size{T}$ variables. The claim follows by marginalizing over the the variables that are not in $\mathcal{V}$.
\end{proof}
Setting $\mathcal{V} = \{  i \}$ for each $i \in  [\size{T}]$ and appealing to the fact that
\[
 \sum_{z \in [\ell]} f_{d+1}(z) f_{T_i \given d+1} (v_i \given z ) =   \sum_{z \in [\ell]} f_{\{T_i,d+1\}}((v_i,z))= f_{T_i} (v_i)~,
\]
we deduce the following corollary.
\begin{corollary} \label{cor:correctness2}
For all subcubes $T$,  
\begin{align*}
& \sum_{z \in [\ell]} f_{d+1}(z) \prod_{i=1}^{\size{T}} f_{T_i \given d+1} (v_i \given z ) \geq \lambda  
 \implies
f_{T_i}(v_i) \geq \lambda
 \end{align*}
for all $i \in [\size{T}]$.
\end{corollary} 

Therefore, we only need to compute $f_{i \given d+1}(x \given z)$ for all coordinates $i \in [d]$, values $z \in [\ell]$ if $x$ is a heavy hitter of coordinate $i$. Similar to the previous section,  for each dimension $i \in [d]$, we find $H_i$ in the first pass and use $H_i$ to find $S_i$ in the second pass.  Appealing to Corrollary \ref{cor:correctness2}, we deduce that if 
\[
q(v) := \sum_{z \in [\ell]} f_{d+1}(z) \prod_{i=1}^{\size{T}} f_{T_{i} \given d+1} (v_i \given z)  \geq  \lambda
 \] 
then  for all $i = 1,2,\ldots,\size{T}$, we have 
$ f_{T_i } (v_i ) \geq \lambda$ which in turn implies that $v_i \in S_{T_i}$. Therefore, we output YES to ${\rm Query}(T,v)$ if and only if all $v_i \in S_{T_i}$ and $q(v) \geq  \lambda$.

To this end, we only need to compute $f_{i \given d+1} (x \given z)$ and $f_{d+1}(z)$ for all $x \in H_i$, $z \in [\ell]$, and $i \in [d]$. The detailed algorithm is as follows.
\begin{tcolorbox}
\begin{enumerate}
\item First pass:
\begin{enumerate}
\item For each value $z \in [\ell]$, compute $f_{ d+1 }(z)$ exactly. 
\item For each coordinate $i \in [d]$, use Misra-Gries algorithm to find $H_i$.
\end{enumerate}
\item Second pass:
\begin{enumerate}
\item For each coordinate $i \in [d]$ and each value $x \in H_i$, compute $f_i (x)$ exactly to obtain $S_i$. 
\item For each value $z \in [\ell]$, coordinate $i \in [d]$, and $x \in H_i$, compute $f_{i \given d+1}(x \given z)$ exactly. 
\end{enumerate}
\item Output YES to ${\rm Query}(T,v)$ if and only if $v_i \in S_{T_i}$ for all $i \in [\size{T}]$ and
\[
 \sum_{z \in [\ell]} f_{d+1}(z) \prod_{i=1}^{\size{T}} f_{T_i \given d+1} (v_i \given z ) \geq \lambda ~.
\]
\end{enumerate}
\end{tcolorbox}

\begin{theorem}
There exists a 2-pass algorithm that uses $\tO(\ell d \gamma^{-1})$ space and solves subcube heavy hitters under the Naive Bayes assumption. The time to answer ${\rm Query}(T,v)$ and  ${\rm AllQuery}(T)$  are  $\tO(\ell k)$ and $O(  \ell (k/ \gamma)^{2} )$  respectively where $k$ is the dimensionality of $T$.
\end{theorem}
\begin{proof}
The space to obtain $H_i$ and $S_i$ over the two passes  is $\tO(d\lambda^{-1})$. Additionally, computing $ f_{i \given d+1}( x \given z)$ for all $i \in [d]$, $z \in [\ell]$, and $x \in H_i$ requires $\tO(\ell d \lambda^{-1})$ bits of space. The overall space we need is therefore $\tO(\ell d \lambda^{-1})=\tO(\ell d \gamma^{-1})$.

The correctness of answering an arbitrary ${\rm Query}(T,v)$ follows directly from Corollary \ref{cor:correctness2}. Specifically, if 
\begin{align}\label{eq:2}
\sum_{z \in [\ell]} f_{d+1}(z) \prod_{i=1}^{\size{T}} f_{T_i \given d+1} (v_i \given z ) \geq \lambda ~,
\end{align}
then, $v_i \in S_{T_i} \subseteq H_{T_i}$ for all $i \in [\size{T}]$ as argued. Hence, $f_{T_i \given d+1} (v_i \given z)$ is computed exactly in the second pass for all $z \in [\ell]$. As a result, we could verify the inequality and output YES. On the other hand, if Eq. \ref{eq:2} does not hold. Then, if some $v_i \notin S_{T_i}$, we will correctly output NO. Otherwise if all $v_i \in S_{T_i}$, then we can compute the left hand side and verify that Eq. \ref{eq:2} does not hold (and correctly output NO).

Obviously, ${\rm Query}(T,v)$ takes $\tO(\ell k)$ time for a $k$-dimensional subcube $T$. We now exhibit a fast algorithm to answer ${\rm AllQuery}(T)$  for a $k$-dimensional subcube $T$. 
 Define
\[ W_j := \{ v \in [n]^j : \sum_{z \in [\ell]} f_{d+1}(z) \prod_{i=1}^j f_{T_i \given d+1}(v_i \given z) \geq \lambda \} ~. \]
Recall that the goal is to find $W_k$. We note that $W_1= S_1$ is obtained directly by the algorithm. Next, we show how to obtain $W_{j+1}$ in $\tO(\lambda^{-2})$ time from $W_j$. Note that $|W_j | \leq 5/(4  \lambda)$ because if $y \in W_j$, then
\begin{align*} 
 \sum_{z \in [\ell]} f_{T_1 \given d+1}(y_1 \given z) \cdots f_{T_j \given d+1}(y_j \given z) f_{d+1}(z)  \geq \lambda \\
\end{align*}
and hence  $f_{{T_{[j]}}} (y) \geq \lambda-\alpha = 4/5 \cdot \lambda^{-1}$   according to the Naive Bayes assumption. This implies that $|W_j | \leq 5/(4 \lambda)$.
  
For each $ (v_1,\cdots,v_j)$ in $W_j$,  we collect all  $v_{j+1} \in S_{j+1}$ such that 
\[
 \sum_{z \in [\ell]} f_{T_1 \given d+1}(v_1 \given z) \cdots f_{T_{j+1} \given d+1}(v_{j+1} \given z) f_{d+1}(z) \geq \lambda
\]
and  put $ (v_1,\cdots,v_{j+1})$ to $W_{j+1}$. Since $|W_j | \leq 5/(4 \lambda)$ and $|S_{j+1}| \leq 1/\lambda$, this step obviously takes $\tO(\ell k \lambda^{-2})$ time. Since we need to do this for $j = 2,3,\ldots,k$, we attain $W_{k}$ in $\tO( \ell (k/\gamma)^{2})$ time. The correctness of this procedure follows directly from Lemma \ref{clm:subspace2} and induction since  $(v_1,\ldots,v_{j+1}) \in W_{j+1}$ implies that $(v_1,\ldots,v_{j})$ is in $W_{j}$ and $v_{j+1}$ is in $S_{j+1}$. Since we check all possible combinations of $(v_1,\ldots,v_{j}) \in W_{j}$ and $v_{j+1} \in S_{j+1}$, we guarantee to construct $W_{j+1}$ correctly.
\end{proof}

%% file: experiments.tex

\section{Experimental study}

\paragraph{Overview.} We experiment {with} our algorithms on a synthetic dataset generated from a Naive Bayes model, and two real-world datasets from Adobe Marketing Cloud\footnote{http://www.adobe.com/marketing-cloud.html} and Yandex. We thoroughly compare the following approaches:
\begin{itemize}
\item The sampling method (\approach{Sampling}) in Section \ref{sec:sampling}.
\item The 2-pass algorithms (\approach{TwoPassAlg}) described in Section \ref{sec:near-independence} and \ref{sec:naive-bayes} depending on the context of the experiment.
\item  The Count-Min sketch heuristic (\approach{Heuristic}): this heuristic uses Count-Min sketch's point query estimation (see \cite{CormodeM04}) to estimate the frequencies given by the near-independence formula (instead of making a second pass through the stream to compute their exact values). We note that this approach has no theoretical guarantee. 
\end{itemize}
We highlight the main differences between the theoretical algorithms in Sections \ref{sec:near-independence} and \ref{sec:naive-bayes} and the actual implementation:
\begin{itemize}
\item Instead of running our algorithms with the theoretical memory bounds, we run and compare them for different memory limits. This approach is more practical and natural from the implementation perspective.
\item In theory, \approach{Sampling} and \approach{TwoPassAlg} use a fixed threshold $\gamma^* = \gamma/2$ to decide between outputting YES or NO. We however experiment with different values of $\gamma^*$ which is helpful when the memory is more limited or when the assumptions are not perfect in real data. 
\end{itemize}
The heavy hitters threshold $\gamma$ is carefully chosen so that the proportion of the number of heavy hitters to the total number joint values to be reasonably small, i.e., approximately at most $ 1\%$ in this paper. Therefore, we use different values of $\gamma$ for each dataset (see  Table \ref{tbl:para_set} for the actual parameter values).

\subsection{Synthetic dataset}
The synthetic dataset is sampled from a pre-trained Naive Bayes model that is used to estimate the probability of a page view. The model was provided by \cite{KvetonBGTMS16} and built on the same Clickstream dataset that we used in Section \ref{sect:clickstream}.
The coordinates consist of one class variable $Z$ and five feature variables $(X_1,\ldots,X_5)$ with high cardinalities. The dataset strongly follows the property that $X_1,\ldots,X_5$ are conditionally independent given $Z$. Specifically, the variables and their corresponding approximate cardinalities are: country (7), city (10,500), page name (8,500), starting page name (6,400), campaign (3,500), browser (300)
where {country} is the class variable.

\paragraph{Warm up experiment} We first evaluate \approach{Sampling} and \approach{Heuristic} on this synthetic dataset. As mentioned earlier, we compare the performance of the two approaches for each fixed memory size.\footnote{We compute the memory use by the  \approach{Sampling} as the product of dimension and the sample size. The memory used by \approach{Heuristic} is computed as the product of dimension and the Count-Min sketch's size. }
We take a subset of approximately {$135,000$} records conditioned on a fixed and most frequent value of $Z$ so that $X_1,\ldots,X_5$ are independent in this subset.
We then run experiments on three different subcubes: $\{X_1,X_2,X_3\}$, $\{X_2,X_3,X_4\}$, and $\{X_3,X_4,X_5\}$.
In this warm up experiment, the main goal is not to find the heavy hitters but to compare the accuracy of the heavy hitters frequency estimations given by \approach{Heuristic} and  \approach{Sampling}. We measure the performance via the mean square error (MSE), the {mean absolute error (MAE), and the mean absolute percentage error (MAPE)}. To do this,  the true frequencies were pre-computed. We use the frequencies of the top 10 heavy hitters in each of the above subcubes. The results (see Figure \ref{fig:ecml2016_pq}) indicate that  \approach{Heuristic} outperforms \approach{Sampling} when restricted to small memory. This warm up experiment gives evidence that knowing the underlying distribution structure helps improving small space heuristic's performance in estimating the heavy hitters frequencies. 

\begin{figure*}[h]
  \centering
  \includegraphics[width=0.85\textwidth,height=3.8cm]{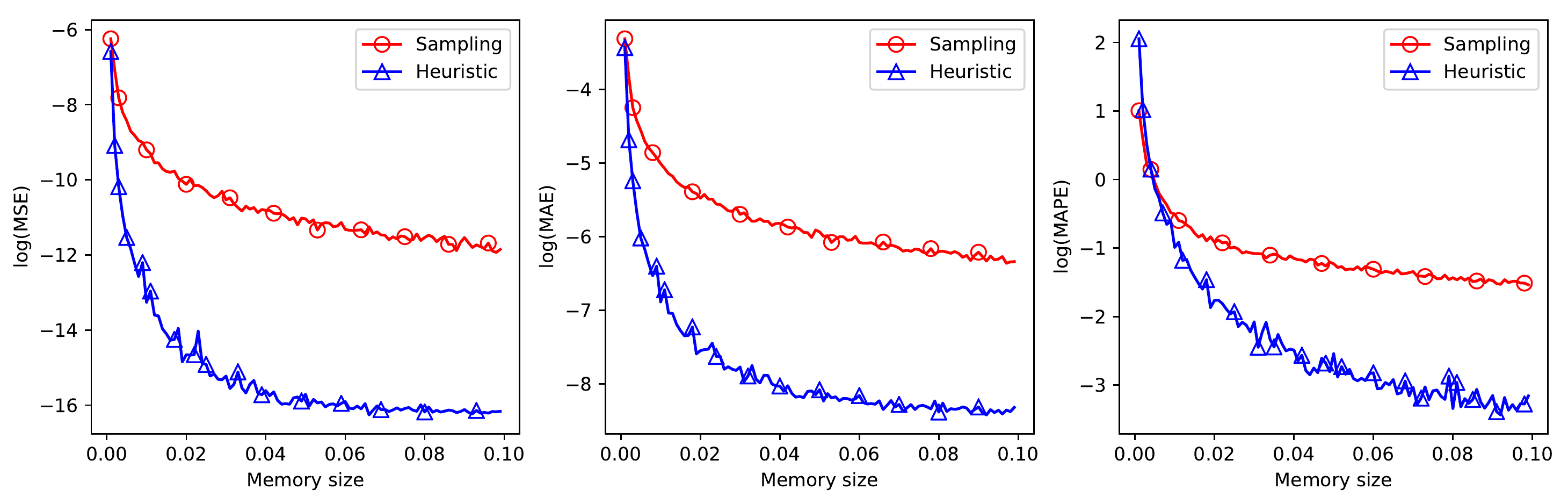}
  \caption{Warm up experiment on synthetic data. Memory size ranges from 0.1\% to 10\% of data size. We report the error as a function of memory size. }
  \label{fig:ecml2016_pq}
\end{figure*}

\paragraph{Experiment with the near-independence assumption.}
We compare performance of the three aforementioned methods on finding heavy hitters under the near-independence assumption. In this experiment, we use the same subset of data and subcubes as in the previous experiment. We fix the memory to be 2\% of data size. 

\small
\begin{table}[t]
\centering
\begin{tabular}{ c c c c c c c }
\hline
Dataset & Mem. &  \#Subcubes & {$\gamma$} & \#HH &  HH ratio \\
\hline
Synthetic (fixed $Z$) & {2}\%  & 3 & {0.002} & {29.7} & {0.079}\% \\
Synthetic (whole) & {2}\% &   3 & {0.002} & {28.7} & {0.054}\% \\
Clickstream        & {10}\% &  4 & {0.002} & {42.0} & {0.165}\% \\
Yandex  & 0.2\% &  8 & {0.1} & 2.2 & 1.65\% \\
\hline
\end{tabular}
\caption{Parameter values for each experiment. \\
(The columns correspond to memory size relative to the dataset, number of the experimented subcubes, average number of heavy hitters, average percentage of heavy hitters.)}
\label{tbl:para_set}
\end{table}
\normalsize
We measure the performance, for different values of $\gamma^*$, based on the number of true positives and false positives. 
As shown in Figure \ref{fig:ecml2016_indep}, for small memory, {both \approach{Heuristic} and \approach{TwoPassAlg} manage to find more heavy hitters than \approach{Sampling}. In terms of false positives, \approach{TwoPassAlg} beats both \approach{Heuristic} and \approach{Sampling} for smaller space. One possible explanation is that when $\gamma^*$ is small (close to $\gamma$), \approach{TwoPassAlg}, with the advantage of the second pass, accurately estimates frequencies of potential heavy hitters whereas the other two methods, especially \approach{Heuristic}, overestimate the frequencies and therefore report more false positives. For larger $\gamma^*$, false positives become less likely and all three approaches achieve similar performances.
} 
In general, \approach{TwoPassAlg} obtains the best performance as seen in the ROC curve.
\begin{figure*}[h]
  \centering
  \includegraphics[width=0.85\textwidth,height=3.8cm]{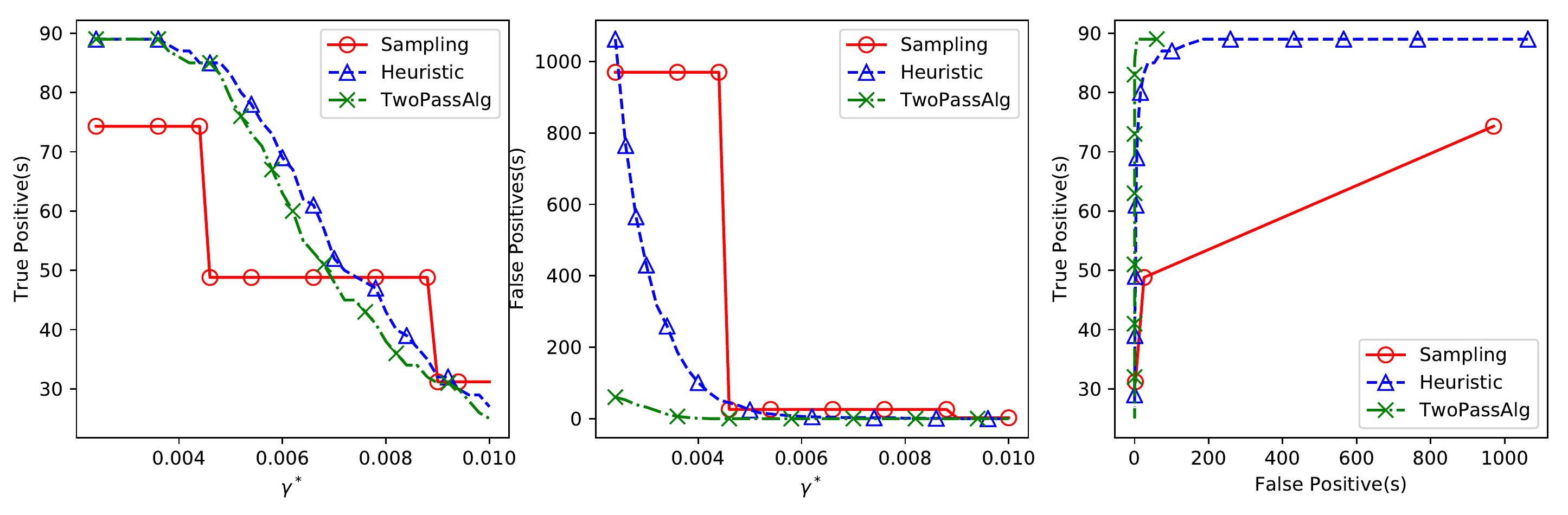}
  \caption{Near-independence experiment on synthetic dataset. We measure the performance based on the number of true and false positives (as a function of $\gamma^\ast$), and the ROC curve.}
  \label{fig:ecml2016_indep}
\end{figure*}
\paragraph{Experiment with the Naive Bayes assumption} 
We use the whole dataset of approximately {$168,000$} records without fixing $Z$ and keep other settings unchanged. We only compared the performance of \approach{TwoPassAlg} and \approach{Sampling} because the conditional probabilities cannot be directly derived from \approach{Heuristic}. In Figure \ref{fig:ecml2016_condindep}, we observe that when restricted to small memory, \approach{TwoPassAlg} attains a better performance by reporting more true heavy hitters and fewer false heavy hitters. As we allow more space, the performance of \approach{Sampling} improves as predicted by our theoretical analysis. 

\begin{figure*}[h]
  \centering
  \includegraphics[width=0.85\textwidth,height=3.8cm]{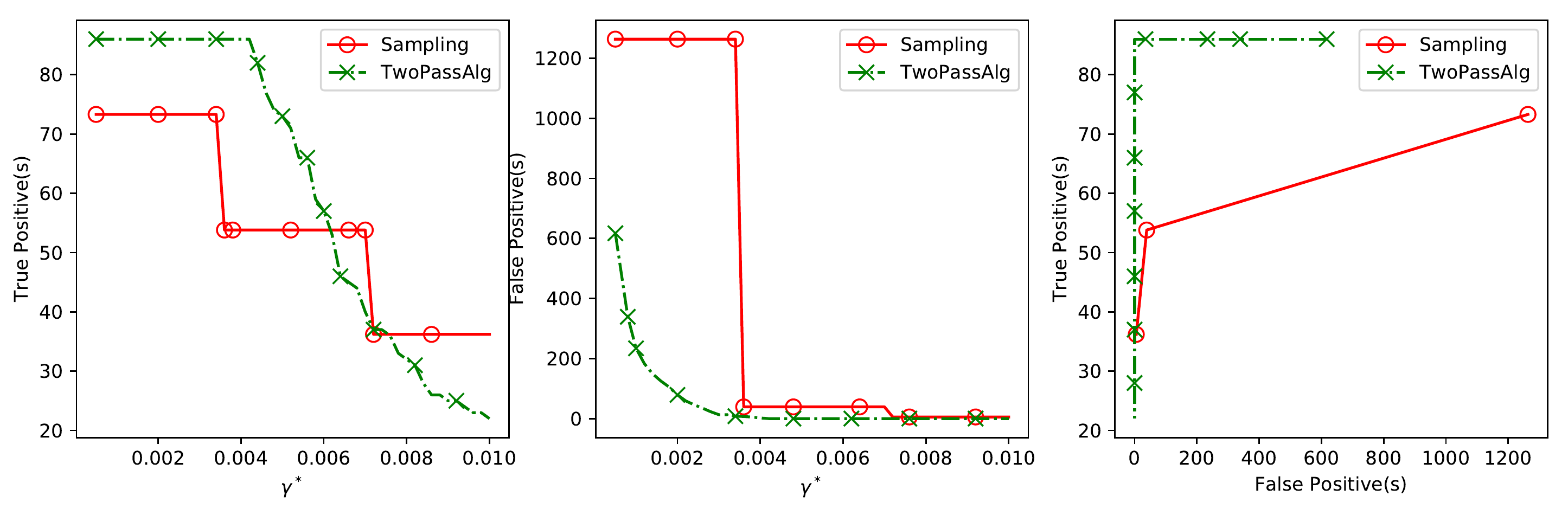}
  \caption{Naive Bayes experiment on synthetic dataset. }
 \label{fig:ecml2016_condindep}
\end{figure*}

\subsection{Clickstream dataset }\label{sect:clickstream}

To evaluate \approach{TwoPassAlg} on real data, we use an advertising dataset called {\em Clickstream Data Feeds} from Adobe Marketing Cloud. 
{The approximate dataset size is {$168,000$} and all values have been anonymized in advance.}

There are 19 high cardinality variables grouped by categories as follows: geography info (city, region, country, domain, carrier), page info ({page name, start page name, first-hit page name}), search info (visit number, referrer, campaign, keywords, search engine), external info (browser, browser width/height, plugins, language, OS).

We avoid obvious correlated features in the query subcubes, e.g., ``search engine'' and ``keywords'' are highly correlated. For example, some highly correlated variables and their correlations are: {start page name \& first-hit page name (0.67), browser \& OS (0.40), region \& country  (0.32), search engine \& country (0.27).}

We carefully select a subset of coordinates that may follow the near independence assumption to query on. For instance, we show our experiment results for the following subcubes, along with the number of heavy hitters recorded:
\{{region}, {page name} , {language}\},
\{{region}, {campaign}, {plugins}\},  
\{{carrier}, {first-hit page name}, {plugins}\},
\{{carrier}, {keywords}, {OS} \}.

Since strong independence property is not guaranteed in this real dataset, we increase memory size to 10\% of the data size in order to obtain better estimation for all methods. Recall that the memory used by \approach{Sampling} and \approach{TwoPassAlg} is partially determined by the number of dimensions and therefore it is reasonable to use a relatively larger memory size. 

In this experiment, all three algorithms are able to find most true heavy hitters (see Figure \ref{fig:market_indep}), but \approach{TwoPassAlg} returns far fewer false positives than the other two methods when $\gamma^*$ is small. In addition, \approach{TwoPassAlg} reaches zero false positive for reasonably large $\gamma^*$.
We can see in the ROC curve that \approach{TwoPassAlg} performs slightly better than \approach{Heuristic} and much better than \approach{Sampling} for small space.

\begin{figure*}[h]
  \centering
  \includegraphics[width=0.85\textwidth,height=3.8cm]{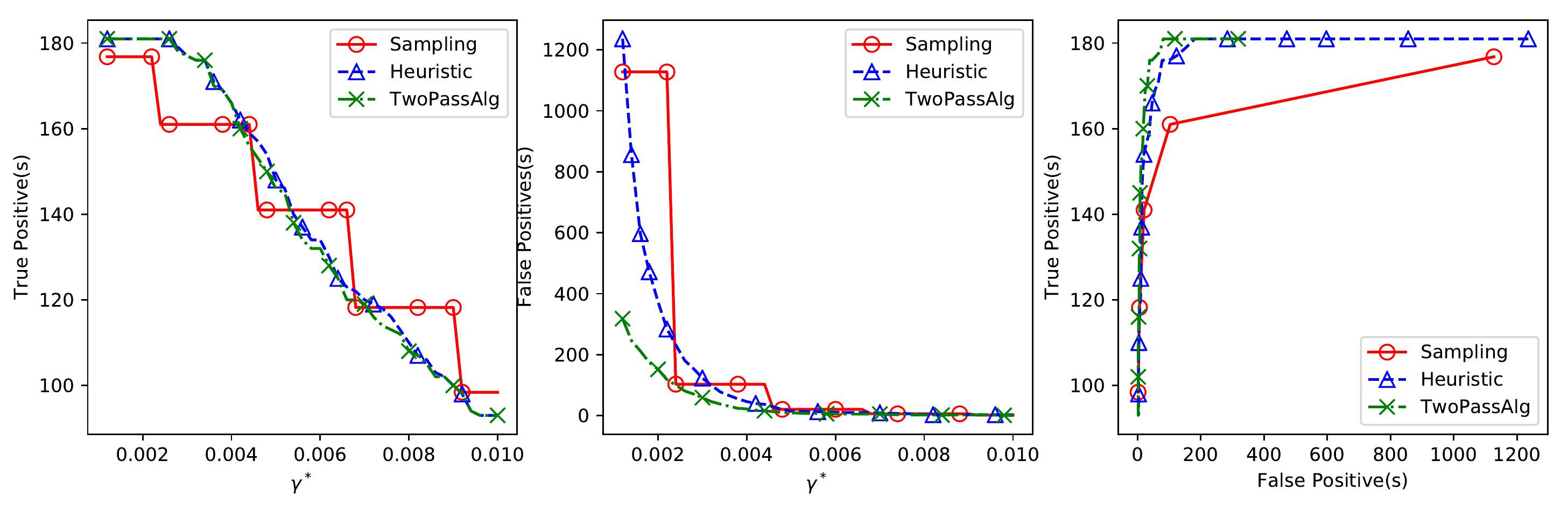}
  \caption{Near-independence experiment with Clickstream dataset. }  
  \label{fig:market_indep}
\end{figure*}

\begin{figure*}[h]
  \centering
  \includegraphics[width=0.85\textwidth,height=3.8cm]{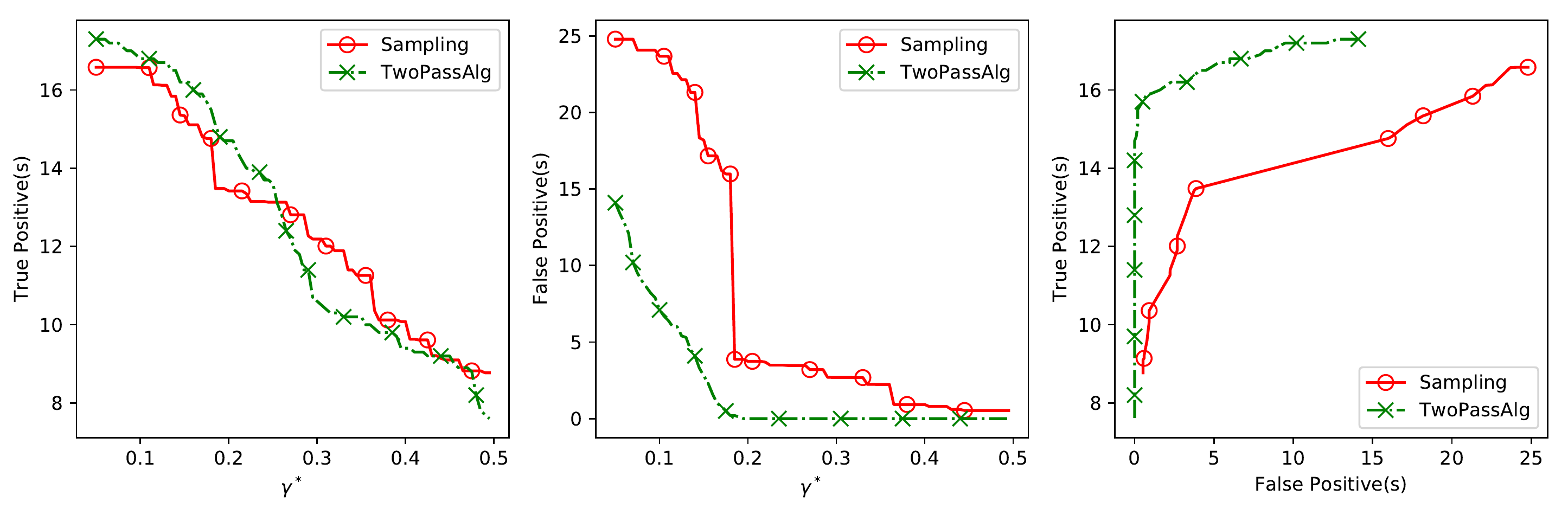}
  
  \caption{Naive Bayes experiment with Yandex dataset. }
  \label{fig:yandex_condindep}
\end{figure*}

\subsection{Yandex dataset} 
Finally, we consider the {\em Yandex dataset}  \cite{yandex} which is a web search dataset (with more than 167 millions data points). Each record in the dataset contains a query ID, the corresponding date, the list of 10 displayed items, and the corresponding click indicators of each displayed item. In the pre-processing step, {we extracted 10 subsets from the whole dataset according to top 10 frequent user queries. The sizes of subsets range from {$61,000$ to $454,000$}.}
In each subset, we treat the first 10 search results as variables $X_1,..,X_{10}$ and ``day of query'' as the latent variable $Z$. 
{
We conjecture that the search results $X_1,..,X_{10}$ are approximately independent conditioned on a given day $Z$. We observe that web patterns typically experience heavy weekly seasonality and these search results largely depend on user query time for some fixed query.
}
We proceed to evaluate \approach{TwoPassAlg} under the Naive Bayes assumption on this dataset.

We consider 8 subcubes in the form $\{X_i,X_{i+1},X_{i+2}\}$ and deliberately set a smaller memory size for this experiment because the cardinality of this dataset is relatively low. We note that different subsets of data may have different number of heavy hitters, so we take the average over 10 subsets as the final result. 

We report the results in Figure \ref{fig:yandex_condindep}. We observe that both \approach{Sampling} and  \approach{TwoPassAlg} are able to find most true heavy hitters. However, \approach{TwoPassAlg} performs significantly better in terms of false positives.

%% file: concl.tex
\section{Concluding Remarks}
Our work demonstrates the power of model-based approach for analyzing high dimensional data that abounds in digital analytics applications. We exhibit algorithms, with fast query time, that overcome worst case space lower bound under the classical Naive Bayes assumption. Our approach to subspace heavy hitters opens several directions for further study.  For example, 
\begin{itemize}
\item Can heavy hitters be detected efficiently under more general models?
\item Can these models be learned or fitted over data streams with polylogarithmic space? We believe this is an algorithmic problem of great interest and will have applications in machine learning beyond the context here. 
\item We assumed that we observe the latent dimension. Can this be learned from the data stream? 
\item Can the model-based approach be extended to other problems besides heavy hitters, including clustering, anomaly detection, geometric problems and others which have been studied in the streaming literature. 
\end{itemize}